\documentclass[submission,copyright,creativecommons]{eptcs}
\providecommand{\event}{CREST 2016} 
\usepackage{breakurl}             
\usepackage{graphicx}
\usepackage{amssymb}
\usepackage{amsmath}
\usepackage{url}

\newenvironment{proof}[1][Proof]{\begin{trivlist}
\item[\hskip \labelsep {\bfseries #1}]}{\end{trivlist}}

\newtheorem{lemma}{Lemma}[section]
\newtheorem{theorem}{Theorem}[section]
\newtheorem{proposition}{Proposition}[section]
\newtheorem{definition}{Definition}[section]

\newcommand{\qed}{\nobreak \ifvmode \relax \else
      \ifdim\lastskip<1.5em \hskip-\lastskip
      \hskip1.5em plus0em minus0.5em \fi \nobreak
      \vrule height0.75em width0.5em depth0.25em\fi}

\usepackage{algpseudocode} 
\usepackage{algorithm}

\usepackage[utf8]{inputenc}
\usepackage{imakeidx}

\title{Debugging of Markov Decision Processes (MDPs) Models}
\author{Hichem Debbi
\institute{Department of Computer Science\\
University of M'sila\\
M'sila, Algeria}
\email{hichem.debbi@univ-msila.dz}
}

\def\authorrunning{H.Debbi}
\def\titlerunning{Debugging of Markov Decision Processes (MDPs) Models}

\begin{document}
\maketitle

\begin{abstract}
	In model checking, a counterexample is considered as a valuable tool for debugging. In Probabilistic Model Checking (PMC), counterexample generation has a quantitative aspect. The counterexample in PMC is a set of paths in which a path formula holds, and their accumulative probability mass violates the probability threshold. However, understanding the counterexample is not an easy task. In this paper we address the task of counterexample analysis for Markov Decision Processes (MDPs). We propose an aided-diagnostic method for probabilistic counterexamples based on the notions of causality, responsibility and blame. Given a counterexample for a Probabilistic CTL (PCTL) formula that does not hold over an MDP model, this method guides the user to the most relevant parts of the model that led to the violation.
\end{abstract}
\section{Introduction}

Probabilistic Model Checking (PMC) has appeared as an extension of model checking for analysing systems that exhibit stochastic behaviour. Several case studies in several domains have been addressed from randomized distributed algorithms and network protocols to biological systems and cloud computing environments. These systems are described usually using Discrete-Time Markov Chain (DTMC), Continuous Time Markov Chain (CTMC) and Markov Decision Process (MDP), and verified against properties specified in Probabilistic Computation Tree Logic (PCTL) \cite{Hansson01} or Continuous Stochastic Logic (CSL) \cite{Aziz02}, \cite{Baier03}.

One of the major advantages of model checking over other formal methods its ability to generate a  counterexample when the model falsifies  such specification. A counterexample is an error trace, by analyzing it, the user can locate the source of the error. Unlike the previous methods proposed for conventional model checking that generate the counterexample as a single path ending with bad states representing the failure, the task in PMC is quite different. The counterexample in PMC is a set of evidences or diagnostic paths that satisfy the path formula and their probability mass violates the probability bound. As it is in conventional model checking, the generated counterexample should be small and most indicative to be easy for understanding . In PMC, this task is more challenging, since the counterexample consists of multiple paths and it is probabilistic. In \cite{Aljazzar12}, the authors introduced a heuristic-based search method for generating counterexamples for DTMCs and CTMCs as what they refer to as diagnostic sub-graphs. In complementary work \cite{Aljazzar13}, they proposed an approach for counterexample generation for MDPs based on existing methods for DTMC. Based on these works, they built an open source tool, DiPro \cite{Aljazzar14}, for generating indicative counterexamples for DTMCs, CTMCs and MDPs. Similar to the previous works, \cite{Han17} has proposed the notion of smallest most indicative counterexample that reduces to the problem of finding K shortest paths. Instead of generating path-based counterexamples, the authors in \cite{Wimmer18} have proposed a novel approach based on critical subsystems. Following this work, \cite{Jansen19} proposed the COMICS tool for generating the critical subsystems that induce the counterexamples. Instead of relying on the state space search resulted from the parallel composition of the modules, \cite{Wimmer21} suggests to rely directly on the guarded command language used by the model checker, which is more likely and helpful for debugging purpose.

Generating small and indicative counterexamples only is not enough for understanding the error. Therefore, in conventional model checking many works have addressed the analysis of counterexamples to better understand the failure \cite{Groce04},\cite{Beer06},\cite{Ball10}. As it was done in conventional model checking, addressing the error explanation in probabilistic model checking is highly required, especially that probabilistic counterexample consists of  multiple paths instead of single path and it is probabilistic. Debbi \cite{Debbi24} has used the definition of causality \cite{Halpern23} for debugging probabilistic models. It has been used with regression analysis in the aim to estimate the causal effect of the Boolean variables on the violation of the probabilistic property, and how these variables change their values dependently, the thing that can help to understand the behavior of the model. The same definition of causality has also been adapted to event orders for generating fault trees from probabilistic counterexamples , where the selection of traces forming the fault tree are restricted to some minimality condition \cite{Leitner-Fischer2013Proba-26504}. To do so, they proposed the event order logic to reason about boolean conditions on the occurrence of events, where the cause of the hazard in their context is presented as a Event Order Logic (EOL) formula, which is a conjunction of events, and the event are simply actions leading from state to another. In  \cite{Fischer26}, they extended their approach by integrating causality in explicit-state model checking algorithm to give a causal interpretation for sub and super-sets of execution traces. They proved the applicability of their approach to many industrial size PROMELA models. They aimed to extended the causality checking approach to probabilistic counterexamples by computing the probabilities of events combination \cite{Leitner125}, but they still consider the use of causality checking of qualitative PROMELA models. They proposed a symbolic version of causality checking \cite{Beer1} based on bounded model checking (BMC) and SAT solving.

In previous work \cite{Debbi22}, we have addressed the analysis of probabilistic counterexamples for DTMCs and CTMCs using the definition of causality by Halpern and Pearl \cite{Halpern23} and its extension, responsibility \cite{Chockler24}. Markov Decision Process (MDP) is a discrete time probabilistic model similar to DTMC, the only difference is that MDP is non-deterministic through the possible actions that can be taken at each state. Hence, for analyzing counterexamples for MDPs, we have to deal with actions as well. To this end, we adapt the definitions of causality and responsibility to reason about causes, in addition we adapt the definition of blame \cite{Chockler24} to reason about actions. All the definitions are adapted in complementary way. Following that, we introduce an algorithm for identifying the relevant parts of the model that contribute the most to the violation of PCTL properties through computing responsibility and blame. We refer to the output of the algorithm as \textit{diagnoses}. So, compared to the work \cite{Debbi22}, this paper does not deal only with the states and the Boolean variables that are considered as causes, but also takes in consideration the transitions leading to them, and under which actions the transitions are enabled. With all these information in hand, we can easily go-back to the model described in guarded command language of PRISM model checker and locate the commands that contribute the most to the violation. In this paper we will focus on probabilistic safety properties with upper threshold. The properties with lower threshold can be easily transformed to properties with upper threshold \cite{Aljazzar12}, \cite{Han17}.

The rest of this paper is organized as follows. In section 2, we present some preliminaries and definitions. In section 3, we show how the notions of causality, responsibility and blame can be adapted for explaining probabilistic counterexamples for MDPs models, following that an algorithm is presented. Experimental results are presented in section 4. At the end, we present conclusion and future works.

\section{Preliminaries and Definitions}

\begin{definition}
	\textbf{(Markov Decision Process (MDP))} is a tuple $M=(S, s_{init}, A, P, L)$, where $S$ is a finite set of states, $s_{init} \in S$ is the initial state, $A$ is a set of actions, $P:S\times A \times S\rightarrow [0,1]$ is a probability transition function such that for every state $s \in S$ and an action $\alpha \in A$ : $\sum_{s' \in S} P(s,\alpha,s') \in \{0,1\}$, and $L: S\rightarrow 2^{AP}$ is a labeling function, which assigns to each state $s\in S$ a subset of the finite set of atomic propositions $AP$.
\end{definition}

At each state $s$, the probability of moving to a successor state $s'$ by taking an action $\alpha$ is given by $P(s,\alpha,s')$. We say that an action $\alpha$ is enabled in state $s$ if and only if $\sum_{s' \in S} P(s,\alpha,s')=1$, otherwise the action $\alpha$ is disabled. For each state $s\in S$  there is at least one action enabled. We denote the set of actions enabled from a state $s$ by $A(s)$. 

An infinite path $\sigma$  is an infinite sequence $\sigma=s_{0}\xrightarrow{\alpha_{0}} s_{1}\xrightarrow{\alpha_{1}}s_{2}...$ with $\alpha_{i} \in A(s_{i})$ such that $P(s_{i},\alpha_{i} ,s_{i+1})$ $>0$ for all $i\geq 0$ . We define the set of infinite paths starting from a state  $s_{0}$ by $Paths(s_{0})$. A finite path is a finite prefix of an infinite path. We denote by $FinitePaths(s_{0})$ the finite paths starting from a state $s_{0}$. For Discrete-time Markov chains (DTMCs), the underlying $\sigma$-algebra is formed by the cylinder sets which are induced by $FinitePaths(s_{0})$. For MDPs, computing the probabilities of paths must rely on the resolution of non-determinism, which is performed by a scheduler. A scheduler $d$ resolves the non-determinism by taking in each state one of the enabled actions $\alpha \in A(s)$, thus resulting in DTMC for which the probability of paths is measurable. We refer to the set of infinite paths under this scheduler as $Paths_{d}(s_{0})$. Then, the underlying $\sigma$-algebra is formed by the cylinder sets which are induced by finite paths under this scheduler denoted $FinitePaths_{d}(s_{0}) $. The probability of this cylinder set is computed by using the following formula:

\begin{equation}
P_{d}(\sigma\in FinitePaths_{d}(s_{0})|\sigma=s_{0}\xrightarrow{\alpha_{0}}s_{1}\xrightarrow{\alpha_{1}}
... \xrightarrow{\alpha_{n-1}} s_{n})=\prod_{i\leq 0 < n}P(s_{i},\alpha_{i},s_{i+1})
\end{equation}

\subsection{Probabilistic Computation Tree Logic (PCTL)}
The Probabilistic Computation Tree Logic (PCTL) has appeared as an extension of CTL for the specification of systems that exhibit stochastic behaviour. PCTL state formulas are formed over the set of atomic propositions $AP$ according to the following grammar:
\begin{flushleft}
	\centering
	$\phi ::=true|a|\neg \phi |{{\phi }_{1}}\wedge {{\phi }_{2}}|{{\textbf{P}}_{\sim p}}(\varphi)$\\
\end{flushleft}
Where $a\in AP$ is an atomic proposition, $\varphi $ is a path formula, $\textbf{P}$ is a probability threshold operator, $\sim \in \{<,\le ,>,\ge \}$ is a comparison operator, and $p$ is a probability threshold. The path formulas $\varphi$ are formed according to the following grammar:
\begin{flushleft}
	\centering
	$\varphi ::= {{\phi }_{1}}\textbf{U}{{\phi }_{2}}|{{\phi }_{1}}\textbf{W}{{\phi }_{2}}|{{\phi }_{1}}{{\textbf{U}}^{\le n}}{{\phi }_{2}}|{{\phi }_{1}}{{\textbf{W}}^{\le n}}{{\phi }_{2}}$\\
\end{flushleft}
Where ${{\phi }_{1}}$ and ${{\phi }_{2}}$ are state formulas and $n\in N$. The PCTL formula is a state formula, where path formulas only occur inside the operator $\textbf{P}$.

The satisfaction of ${{\textbf{P}}_{\sim p}}(\varphi)$ on DTMC depends on the probability mass of set of paths satisfying $\varphi$. This set is considered as a countable union of cylinder sets, so that, its measurability is ensured. A formula ${{\textbf{P}}_{\sim p}}(\varphi)$ is satisfied on an MDP $M$ if only if for every $d\in D$: ${{\textbf{P}_{d}}}(\varphi){\sim p}$, where $D$ represents the set of all schedulers and ${{\textbf{P}_{d}}}(\varphi)$ represents the probability of the set of all finite paths satisfying $\varphi$ under the scheduler $d$.

The semantics of PCTL state and of path formulas for MDPs are defined as the same as for DTMCs  as follows.
\begin{flushleft}
	$s\models true \Leftrightarrow  true$\\
	$s\models a \Leftrightarrow a\in L(s)$\\
	$s\models\neg \phi \Leftrightarrow s\not\models \phi$\\ 
	$s\models{{\phi }_{1}}\wedge {{\phi }_{2}} \Leftrightarrow s\models{{\phi }_{1}}\wedge s\models{{\phi }_{2}}$\\
\end{flushleft}
Given a path $s_{0}\xrightarrow{\alpha_{0}}s_{1}\xrightarrow{\alpha_{1}}...$ and an integer $j\geq 0$, where $\sigma[j]=s_{j}$, the semantics of PCTL path formulas is defined as for CTL as follows:
\begin{flushleft} 
	$\sigma \models{{\phi }_{1}}\textbf{U}{{\phi }_{2}} \Leftrightarrow \exists j\ge 0.\sigma \left[ j \right]\models{{\phi }_{2}}\wedge (\forall 0\le k<j.\sigma \left[ k \right]\models{{\phi }_{1}}) $\\ 
	$\sigma \models{{\phi }_{1}}\textbf{W}{{\phi }_{2}} \Leftrightarrow \sigma \models{{\phi }_{1}}\textbf{U}{{\phi }_{2}} \vee (\forall k\ge 0.\sigma \left[ k \right]\models{{\phi }_{1}})$\\ 
	$\sigma \models{{\phi }_{1}}{{\textbf{U}}^{\le n}}{{\phi }_{2}} \Leftrightarrow \exists 0\le j\le n.\sigma \left[ j \right]\models{{\phi }_{2}}\wedge (\forall 0\le k<j.\sigma \left[ k \right]\models{{\phi }_{1}}) $\\ 
	$\sigma \models{{\phi }_{1}}{{\textbf{W}}^{\le n}}{{\phi }_{2}} \Leftrightarrow \sigma \models{{\phi }_{1}}\textbf{U}^{\le n}{{\phi }_{2}} \vee (\forall 0\le k\le n.\sigma \left[ k \right]\models{{\phi }_{1}})$\\
\end{flushleft}

We should mention that for model checking of MDPs we have to consider  either maximizing or minimizing scheduler. Let $P_{max}(\varphi)$ be the maximal probability of $\varphi$ where $P_{max}(\varphi) = max\{P_{d}(\varphi)| d\in D\}$, and dually the minimal probability $P_{min}(\varphi)$ be the minimal probability of $\varphi$ where $P_{min}(\varphi) = min\{P_{d}(\varphi)| d\in D\}$. For instance for properties of upper threshold , it is evident that $( M \not\models P_{\leq p}(\varphi)) \Leftrightarrow P_{max}(\varphi) > p$.

\subsection{Probabilistic Counterexamples}
The PCTL property $\phi ={{\textbf{P}}_{\le p}}(\varphi)$ is violated on an MDP, if there exists at least one scheduler $d$ such that the probability mass of the paths satisfying $\varphi$ under $d$ exceeds the bound $p$. A probabilistic counterexample for the property $\phi ={{\textbf{P}}_{\le p}}(\varphi)$ can be formed of a set of paths from $FinitePaths_{d}(s_{0})$ starting at state $s_{0}$ and satisfying the path formula $\varphi$ such that $P_{max}(\varphi) > p$ for some scheduler $d$. We denote this set by $FinitePaths_{d}(s_{0}\models\phi)$. These finite paths are also called diagnostic paths \cite{Aljazzar12},\cite{Aljazzar13}.

It is clear that given a scheduler $d$, it is possible to find a set of probabilistic counterexamples under $d$ denoted $PCX_{d}(s_{0}\models\phi )$, which is a set of any combination from  $FinitePaths_{d}(s_{0}\models\phi)$, their probability mass exceeds the bound $p$. Among all these probabilistic counterexamples, we are interested by the most indicative one. A most indicative counterexample is a minimal counterexample (has the least number of paths from $FinitePaths_{d}(s_{0}\models\phi)$) and its probability mass is the highest among all other minimal counterexamples. We denote a most indicative probabilistic counterexample by $MIPCX_{d}({{s}_{0}}\models\phi )$. We should mention that a most indicative probabilistic counterexample may not be unique.

\begin{lemma}\label{lem:Lema1}
	For every path $\sigma \in MIPCX_{d}({{s}_{0}}\models\phi )$ on which the property ${{\phi }_{1}}\textbf{U}{{\phi }_{2}}({{\phi}_{1}}{{\textbf{U}}^{\le n}}{{\phi}_{2}})$ is satisfied, the right state sub-formula ($\phi_2$) is satisfied in the last state of $\sigma$.
\end{lemma}
\textbf{Example 1}
Let us consider the example of MDP shown in Figure 1 and the property $P_{\leq0.5}(\varphi)$, where $\varphi=(a \vee b) U (c\wedge d )$.

\begin{figure}
	\centering
	\includegraphics[width=3.5 in]{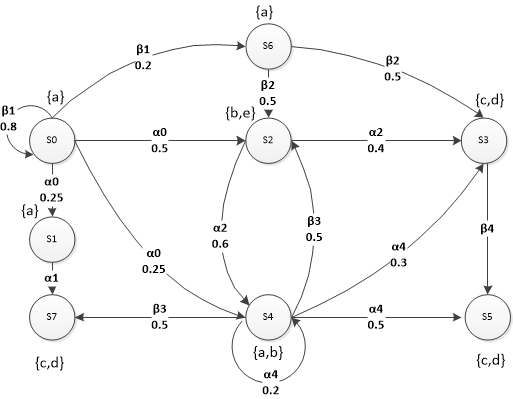}
	\caption{An Example for an MDP}
	\label{fig:MDP}
\end{figure}
This property is violated in this model $( s_{0}\not\models P_{\leq 0.5}(\varphi))$, since there exists a scheduler $d$ (of $\alpha$-actions) that induces a set of finite paths satisfying $\varphi$ and their probability mass is greater than the probability bound (0.5). Any combination from $FinitePaths_{d}(s_{0}\models\phi)$  having probability mass greater than 0.5, is a valid probabilistic counterexample including the whole set. Let us take the following counterexamples:
\begin{flushleft}
	$P(CX_{1})=P(\{s_{0}\xrightarrow{\alpha_{0}}s_{1}\xrightarrow{\alpha_{1}}s_{7}, s_{0}\xrightarrow{\alpha_{0}}s_{2}\xrightarrow{\alpha_{2}}s_{3},s_{0}\xrightarrow{\alpha_{0}}s_{2}\xrightarrow{\alpha_{2}}s_{4}\xrightarrow{\alpha_{4}}s_{3},s_{0}\xrightarrow{\alpha_{0}}s_{2}\xrightarrow{\alpha_{2}}s_{4}\xrightarrow{\alpha_{4}}s_{5},s_{0}\xrightarrow{\alpha_{0}}s_{4}\xrightarrow{\alpha_{4}}s_{5}$
	$\})$\\
	$~~~~~~~~= 0.25 + 0.2 + 0.09 + 0.15 + 0.12 =  0.81$
\end{flushleft}

\begin{flushleft}
	$P(CX_{2})=P(\{{s_{0}\xrightarrow{\alpha_{0}}s_{1}\xrightarrow{\alpha_{1}}s_{7},s_{0}\xrightarrow{\alpha_{0}}s_{2}\xrightarrow{\alpha_{2}}s_{4}\xrightarrow{\alpha_{4}}s_{5},
		s_{0}\xrightarrow{\alpha_{0}}s_{4}\xrightarrow{\alpha_{4}}s_{5}}\})$\\
	$~~~~~~~~= 0.25 + 0.15 + 0.12 = 0.52$
\end{flushleft}
\begin{flushleft}
	$P(CX_{3})=P(\{{s_{0}\xrightarrow{\alpha_{0}}s_{1}\xrightarrow{\alpha_{1}}s_{7},s_{0}\xrightarrow{\alpha_{0}}s_{2}\xrightarrow{\alpha_{2}}s_{3},s_{0}\xrightarrow{\alpha_{0}}s_{2}\xrightarrow{\alpha_{2}}s_{4}\xrightarrow{\alpha_{4}}s_{5}\})}$\\
	$~~~~~~~~= 0.25 + 0.2 + 0.15 =0.60$
\end{flushleft}

The last probabilistic counterexample $CX_{3}$ is most indicative since it is minimal and its probability is greater than the other minimal counterexample $CX_{2}$.

\section{Generating Diagnoses for MDPs}

\subsection{Causality and Responsibility for $MIPCX_{d}({{s}_{0}}\models\phi )$}
As the same as for DTMCs, explaining the violation of PCTL properties of the form $\phi ={{\textbf{P}}_{\leq p}}(\varphi )$ for MDPs reduces to the explanation of exceeding the probability bound over the MDP model. MDP is a discrete time probabilistic model similar to DTMC. Therefore, the definitions of causality and responsibility can be adapted for  $MIPCX_{d}({{s}_{0}}\models\phi )$ in the same way they have been adapted for counterexamples of DTMCs \cite{Debbi22}. 
\begin{definition}
	\textbf{(Causality Model)}
	A causality model for a most indicative probabilistic counterexample  $MIPCX_{d}({{s}_{0}}\models\phi )$ is a tuple $M=(U,V,F)$, where the set $U$ is represented by a context variable; its value $u$ represents a state $s$. $V$ is a set of atomic propositions and Boolean formulas. $F$ associates with every variable $V_{i}\in V$ a truth function $f_{V_{i}}$ that determines the value of  $V_{i}$ (0 or 1), given a state $s$ and the values of other variables in $V$.
\end{definition}
Let us denote by $\widehat{MIPCX_{d}{(s,X\leftarrow x')}({{s}_{0}}\models\phi )}$ the  set of finite paths resulted from  $MIPCX_{d}({{s}_{0}}\models\phi )$ by switching the value of a variable $X\in V$ in a state $s$.
\begin{definition}
	\textbf{(Criticality)}
	Consider a counterexample $MIPCX_{d}({{s}_{0}}\models\phi )$ for a probabilistic formula $\phi ={{\textbf{P}}_{\leq p}}(\varphi )$, a state $s$ in $MIPCX_{d}({{s}_{0}}\models\phi)$ and a variable $X \in V$ that has a value $x \in \{0,1\}$ in $s$. We say that a pair $(s,X=x)$ is critical for the violation of $\phi ={{\textbf{P}}_{\leq p}}(\varphi )$, if  $\widehat{MIPCX_{d}{(s,X\leftarrow x')}({{s}_{0}}\models\phi )}$ is not a valid counterexample for $\phi ={{\textbf{P}}_{\leq p}}(\varphi )$.
\end{definition}

\begin{definition}
	\textbf{(Actual Cause)}
	Consider a counterexample  $MIPCX_{d}({{s}_{0}}\models\phi )$ for a probabilistic formula $\phi ={{\textbf{P}}_{\leq p}}(\varphi )$ and a variable $X \in V$. We say that $(s,X=x)$  is a cause for the violation of $\phi ={{\textbf{P}}_{\leq p}}(\varphi )$, if $(s,X=x)$  is critical , or there exists a subset of variables $W$ of $V$ such that switching their values in $s$ makes $(s,X=x)$ critical. 
\end{definition}

Thus, in our setting, we can refer to a cause as a pair $(s,X=x)$ where $X$  is an atomic proposition has the value 1 in $s$ if $X\in L(s)$, and it has the value 0 otherwise. $(s,X=x)$ is said to be cause, if it is critical, or it can be made critical by switching the values of a set of variables $W$  in $s$.

\begin{definition}
	\textbf{(Responsibility)}
	Consider a counterexample $MIPCX_{d}({{s}_{0}}\models\phi )$ for a probabilistic formula $\phi ={{\textbf{P}}_{\leq p}}(\varphi )$ and a variable $X \in V$. The degree of responsibility of a cause $(s,X=x)$ for the violation of $\phi ={{\textbf{P}}_{\leq p}}(\varphi )$ denoted $dR(s,X=x,\phi)$is 1 if $(s,X=x)$ is critical, and otherwise is $1/(\lvert W\rvert+1)$. 
\end{definition}

That is, we can think of responsibility as a quantitative measure that gives us diagnostic information on $(s,X=x)$ being a cause for the violation of $\phi ={{\textbf{P}}_{\leq p}}(\varphi )$ , where the cause having the highest responsibility for the violation is the critical one.

\begin{definition}
	\textbf{(Probabilistic Causality Model)}
	A probabilistic causality model for $MIPCX_{d}({{s}_{0}}\models\phi )$ is a tuple $(M,Pr)$, where $M$ is the causality model and $Pr$ is a probability function defined over the states of $MIPCX_{d}({{s}_{0}}\models\phi )$. We define for each state $s$ from $MIPCX_{d}({{s}_{0}}\models\phi )$ its probability as follows:
	\begin{equation}
	Pr(s)=\sum_{\sigma\in MIPCX_{d}({{s}_{0}}\models\phi) | s\in \sigma } P(\sigma)
	\end{equation}
\end{definition}

Since every variable in  $V$ is a function of $U$, we can define the probability of each cause $(s,X=x)$ in the same way:
\begin{equation}
Pr(s,X=x)=Pr(s)
\end{equation}

\begin{definition}
	\textbf{(Most Responsible Cause)}
	A cause $C$ is a most responsible cause for the violation of $\phi ={{\textbf{P}}_{\leq p}}(\varphi )$, if $dR(C)Pr(C) \geq dR(C')Pr(C')$ for any cause $C'$.
\end{definition}

\subsection{Blame for $MIPCX_{d}({{s}_{0}}\models\phi )$}

MDPs are non-deterministic Markov models through the possible actions that can be taken at each state, where each action enables a set of transitions. In probabilistic guarded command language of probabilistic model checker like PRISM, transitions represent updates on the variables with probabilities assigned to them. Therefore, for complete analysis of probabilistic counterexamples for MDPs, we have also to deal with actions. As a result, we should ask: how does an action $\alpha$ contribute to the violation of $\phi ={{\textbf{P}}_{\leq p}}(\varphi )$? which leads to the question, which action should be more blamed for exceeding the probability threshold $p$. Hence, the designer would not need just to be aware of probable false assignments, but he also needs to know how such actions are involved. Given this information, the designer could fix the guard commands in way he gets the acceptable outcome. We should mention that blame considers mainly whether an action $\alpha$ is to blame for an outcome $\varphi$ under uncertainty \cite{Chockler24}.

Consider a state $s$ in $MIPCX_{d}({{s}_{0}}\models\phi )$. For each action $\alpha \in A(s)$, we denote by $Suc(s, \alpha)$ the set of $\alpha$-successors of $s$, where $\alpha$-successor is a state $s' \in S$ such that $P(s,\alpha,s')>0$. We associate for each transition from  $s$ to $s'$, where $s'\in Suc(s, \alpha)$ a probability as follow
\begin{equation}
P^{MIPCX_{d}({{s}_{0}}\models\phi )}_{\alpha}(s,s')=\sum_{\sigma\in MIPCX_{d}({{s}_{0}}\models\phi) | (s,\alpha,s') \in \sigma} P(\sigma)
\end{equation}
It is evident that not every transition enabled by an action $\alpha$ has the same presence in the paths forming the counterexample. So this probability measures simply the contribution of a transition to the probability of the counterexample by summing the probabilities of the paths in which it is included. 
\begin{proposition}\label{prop:Prop1}
	Consider a transition $(s,\alpha,s')$ in $MIPCX_{d}({{s}_{0}}\models\phi )$, $Max (P^{MIPCX_{d}({{s}_{0}}\models\phi )}_{\alpha}(s,s'))= Pr(s)$ iff $s'$ is the unique successor of $s$.
\end{proposition}
\begin{proof}
	$\Pr(s)$ represents the sum of probabilities of paths in which $s$ is included, $P^{MIPCX_{d}({{s}_{0}}\models\phi )}_{\alpha}(s,s'))$ represents the probabilities of the paths in which the transition is included, so it is sufficient to prove that both of them are included in the same set of paths, if only if $s'$ is the unique successor of $s$. Let $s'$ and $s''$ two successors of $s$, if $s$ is included in $N$ paths, the transition $(s,s'')$ will be included at least in one path from this set, and thus $s'$ will be included at most in $N-1$ of paths. Hence, $Max(P^{MIPCX_{d}({{s}_{0}}\models\phi )}_{\alpha}(s,s'))$ will not equal to $Pr(s)$ if there is another successor $s''$ of $s$.
\end{proof}

We should mention that every transition in a probabilistic program describes how the values of the variables evolve over time, and thus considering the transitions and their contribution to the error is very important as debugging information, which makes it a required measure for the definition of blame.
\begin{definition}
	\textbf{(Blame)}
	Consider a counterexample $MIPCX_{d}({{s}_{0}}\models\phi)$ for a probabilistic formula $\phi ={{\textbf{P}}_{\leq p}}(\varphi )$ , a state $s$, an action $\alpha\in A(s)$ and a set of successors $Suc(s, \alpha)$. The degree of blame for an action $\alpha$ for the violation of $\phi ={{\textbf{P}}_{\leq p}}(\varphi )$ in a state $s$ denoted $dB(s,\alpha,\phi)$ is
	\begin{equation}
	\sum_{s' \in Suc(s, \alpha)}dR(s',X=x,\phi) P^{MIPCX_{d}({{s}_{0}}\models\phi )}_{\alpha}(s,s')
	\end{equation} 
\end{definition}
That is, the degree of blame dB informs us about the contribution of an action $\alpha$ in a state $s$ to the violation of the probabilistic formula $\phi ={{\textbf{P}}_{\leq p}}(\varphi )$. While responsibility stands as a criticality measure for actual causes given well-defined states for the violation of $\phi ={{\textbf{P}}_{\leq p}}(\varphi )$, dB describes how an action should be blamed for this violation through considering the probabilities assigned to the transitions leading to these states. So that, the action more blamed for the violation will be the one more contributing to  the probability of $MIPCX_{d}({{s}_{0}}\models\phi )$ and leading to more critical states.

\begin{definition}
	\textbf{(Most Blame)}
	An action $\alpha \in A(s)$ has most blame for the violation of PCTL property $\phi ={{\textbf{P}}_{\leq p}}(\varphi )$, if $dB(s,\alpha,\phi) \geq dB(s',\alpha',\phi)$ for any $s'$ and any $\alpha' \in A(s')$.
\end{definition}


\begin{proposition}\label{prop:Prop2}
	Let $\alpha \in A(s)$ , $Max(dB(s,\alpha,\phi))=P(MIPCX_{d}({{s}_{0}}\models\phi ))$ iff for every $\sigma \in MIPCX_{d}({{s}_{0}}\models\phi )$ there exists $(s,\alpha,s') \in \sigma$, such that $s'$ is critical with respect to such variable $X$.
\end{proposition}
That is, the maximum degree of blame of an action $\alpha$ in a state $s$ is equal to the probability of the counterexample, if only if for every path $\sigma$ of the counterexample, there exists a transition from $s$ to a state $s'$ under this action, and $s'$ is critical with respect to such variable $X$.
\begin{theorem}
	Let $\alpha \in A(s)$ and $s' \in Suc(s,\alpha)$, $(s',X=x)$ is most responsible cause $\nRightarrow$ $\alpha$ has the most blame.
\end{theorem}
Proof. Let $s_{1}$ and $s_{2}$ be two states, $\alpha_{1}$ and $\alpha_{2}$ two actions enabled at these states respectively. Let $\alpha_{1}$ leads to a most responsible cause $C$ and we assume that $dB(s_{1},\alpha_{1},\phi) \geq dB(s_{2},\alpha_{2},\phi)$ . Then, for every cause $C'$, $\alpha_{2}$ leads to, $dR(C')Pr(C')< dR(C)Pr(C)$. Let $C$ be the unique cause $\alpha_{1}$ it leads to, let $Pr(s_{2})>Pr(s_{1})$ and let all the causes that $\alpha_{1}$ and $\alpha_{2}$ lead to are critical. With respect to Proposition ~\ref{prop:Prop1}, the probability of the transition leading to $C'$ is $Pr (s_{1})$, and thus with respect to the definition  of blame, $dB(s_{1},\alpha_{1},\phi)$ will be at most $Pr(s_{1})$, which is less than $Pr(s_{2})$, and since the causes that $\alpha_{2}$ leads to are critical, it contradicts $\alpha_{1}$ being the action with the most blame.

\subsection{Algorithm for Generating Diagnoses }
This algorithm performs on counterexamples generated by DiPro \cite{Aljazzar14}. DiPro is a tool used for generating counterexamples from DTMCs, CTMCs and MDPs models, and can be jointly used with the probabilistic model checkers PRISM \cite{Hinton15} and MRMC \cite{Katoen16}.

\begin{algorithm}
	\caption{. Generate Diagnoses}
	\begin{algorithmic}[1]
		\State \textbf{Inputs:} The counterexample $MIPCX_{d}({{s}_{0}}\models\phi )$,
		The probabilistic formula $\phi ={{P}_{\le p}}(\varphi )$ where $\varphi$ is of
		the form ${{\phi }_{1}}\textbf{U}{{\phi }_{2}}$ or $({{\phi }_{1}}{{\textbf{U}}^{\le n}}{{\phi }_{2}})$
		\State \textbf{Outputs:} Causes with responsibilities and probabilities\\
		~~~~~~~~~~~~~ Actions with blame\\
		\State Causes \textbf{:=} $\emptyset$
		\State Actions \textbf{:=} $\emptyset$
		\For {\textbf{each} state $s$ in $MIPCX_{d}({{s}_{0}}\models\phi)$}
		\State W{:=} 0
		\If{$s$ is the last state in a path $\sigma$} 
		\State Causes with dR\textbf{:=} Causes $\cup$ FINDCAUSES($s,\phi_{2}$, $W$)
		\State Pr(Causes)=$\sum_{\sigma\in MIPCX_{d}({{s}_{0}}\models\phi ) | s\in \sigma} P(\sigma)$
		\Else
		\State Causes with dR\textbf{:=} Causes $\cup$ FINDCAUSES($s,\phi_{1}$, $W$)
		\State Pr(Causes)=$\sum_{\sigma\in MIPCX_{d}({{s}_{0}}\models\phi ) | s\in \sigma} P(\sigma)$
		\State $P^{MIPCX_{d}({{s}_{0}}\models\phi )}_{\alpha}(s,s')=\sum_{\sigma\in MIPCX_{d}({{s}_{0}}\models\phi) | (s,\alpha,s') \in \sigma} P(\sigma)$
		\EndIf
		\EndFor
		
		\For {\textbf{each} $s$ in $MIPCX_{d}({{s}_{0}}\models\phi )$}
		\State Actions with dB\textbf{:=} Order(Actions $\cup$ {$\alpha \in A(s)$, $dB(s,\alpha,\phi_{1})=\sum_{s' \in Suc(s, \alpha)}dR(s',X=$\\
			~~~~~~~~~~~~~~~~~~~~~~~~~~$x,\phi_{1}) P^{MIPCX_{d}({{s}_{0}}\models\phi )}_{\alpha}(s,s')$})
		\EndFor
		
		\State \textbf OUTPUTDIAGNOSES(Causes with dR and Pr, Actions with dB)
	\end{algorithmic}
\end{algorithm}

\begin{algorithm}
		\begin{algorithmic}[2]
			\Function{FindCauses}{$s$, $\psi$, $W$}
			\If{$\psi$ is of the form $a$ where $a\in AP$ and $a\in L(s)$}
			\State \textbf{return} $(\langle s,a \rangle, dR(\langle s,a \rangle)=1/(W+1)$)
			\EndIf
			
			\If{$\psi$ is of the form $\neg a$ where $a\in AP$and $a\notin L(s)$}
			\State \textbf{return} $(\langle s,\neg a \rangle, dR(\langle s,\neg a \rangle)=1/(W+1))$
			\EndIf
			
			\If{$\psi$ is of the form ${{\psi}_{1}}\wedge {{\psi}_{2}}$}
			\State \textbf{return} {FindCauses($s$, $\psi_{1}$, $W$) $\cup$ 
				\State FindCauses($s$, $\psi_{2}$, $W$)}
			\EndIf
			
			\If{$\psi$ is of the form ${{\psi}_{1}}\vee {{\psi}_{2}}$}
			\If {$s\models \psi_{1}$ and $s\models \psi_{2}$ }
			\State \textbf{return} {FindCauses($s$, $\psi_{1}$, $W++$) $\cup$
				\State FindCauses($s$, $\psi_{2}$, $W++$)}
			\If{$s\models \psi_{1} \wedge s\not\models \psi_{2}$}
			\State \textbf{return} {FindCauses($s$, $\psi_{1}$, $W$)}
			\EndIf
			\If{$s\not\models \psi_{1} \wedge s\models \psi_{2}$}
			\State \textbf{return} {FindCauses($s$, $\psi_{2}$, $W$)}
			\EndIf
			\EndIf
			\EndIf
			\EndFunction
		\end{algorithmic}
		
	\begin{algorithmic}[3]
		\Function{OutputDiagnoses}{Causes with dR and Pr, Actions with dB}
		
		\For {\textbf{each} action$\alpha \in A(s)$ from Actions}
		\State Output ($\alpha$)
		\For	{\textbf{each} a successor $s' \in Suc(s,\alpha)$}
		\State Output Causes Ordered with respect to $dR\times Pr$
		\State Output $(s,s')$
		\EndFor
		\EndFor
		\EndFunction
	\end{algorithmic}

\end{algorithm}

The algorithm 1 (Generate Diagnoses) explores the counterexample and computes the causes with their responsibilities and probabilities with respect to each state $s$, and computes the degree of blame for each action enabled at this state. The condition put on last state follows Lemma ~\ref{lem:Lema1}. Algorithm 1 uses the function FindCauses that takes a state and a state formula as input as well as the variable $W$, and returns recursively the set of causes and their responsibilities, where the variable $W$ is used to compute responsibility \cite{Debbi22}.

Computing the set of causes exactly in binary causal models is NP-complete in general\cite{Eiter28}. However, \cite{Eiter29} shows that causes can be computed in polynomial time if the causal graph forms a directed causal tree. Our algorithm is linear in the size of  $MIPCX_{d}({{s}_{0}}\models\phi)$ and the size of $\varphi$ as it has been presented in \cite{Debbi22}, because reconfiguring the algorithm to  compute the degree of blame by adding the loop in line 18 does not bring additional complexity since it is directly based on measures already computed, which are the degree of responsibility of each cause and the probability of each transition $P^{MIPCX_{d}({{s}_{0}}\models\phi )}_{\alpha}(s,s')$. Computing $P^{MIPCX_{d}({{s}_{0}}\models\phi )}_{\alpha}(s,s')$ (line 15) is performed under the same loop for computing the probabilities of causes (line 14).

At the end we present the function OUTPUTDIAGNOSES that shows the way of presenting the diagnoses to te user. It gets the actions ordered with respect to dB, and the causes with dR and Pr and outputs the diagnoses in order. We see that the output of the diagnoses starts with the action with the most blame, and among all the causes it leads to, we begin by the most responsible cause, by indicating also the transition leading to this cause. Presenting the transition enabled under this action is very important as a diagnostic information, because transitions in typical probabilistic program describe how the values of the variables evolve over time.

\textbf{Example 2}\\
Let us apply this algorithm on the counterexample $CX3$ from the previous example. The user gets the action $\alpha_{0}$ first, since $dB(s_{0},\alpha_{0},\phi)=0.6$ is the highest, it is equal to the probability of the counterexample. From $Suc(s_{0},\alpha)$, the user gets first the cause $(s_{2},b)$ because it is the most responsible by computing the measure $Pr(s_{2},b)\times dR(s_{2},b)=0.58 > Pr(s_{1},a)\times dR(s_{1},a,\phi)=0.41$. The following action the user gets is $\alpha_{2}$ with degree of blame $dB(s_{2},\alpha_{2},\phi)=0.5\times(0.15)+1\times(0.2)=0.275$ with the causes it led to $\{(s_{3},c),(s_{3},d)\}$ and $\{(s_{4},a)\}$,$\{(s_{4},b)\}$ respectively, then $\alpha_{1}$ with $dB(s_{1},\alpha_{1})=0.25$ with the cause it led to  $\{(s_{7},c),(s_{7},d)\}$, and finally the action $\alpha_{4}$ with $dB(s_{4},\alpha_{4},\phi)=0.15$ with the cause it led to  $\{(s_{5},c),(s_{5},d)\}$.

\section {Experiments}
We have implemented the above method in Java. We used two benchmark case studies to evaluate our method, the Zeroconf protocol \cite{IP29} and CSMA/CD protocol \cite{CSMA31}. All the experiments were carried out on windows XP with Intel Pentium CPU 3.2 GHz speed And 2 GB of memory. We sue DiPro for generating the counterexamples. DiPro employs many algorithms for generating counterexamples, among these algorithms we use the K* algorithm \cite{Aljazzar20112129}. Our method takes the counterexample generated from DiPro and the property to be verified as input, and outputs the diagnoses. 

\subsection{Zeroconf}
The protocol is modeled in PRISM as an MDP, where the number of abstract hosts is denoted by $N$, the number  of probes to send is denoted by $K$, and the probability of message loss is denoted by $loss$. Each station has a single-message buffer and is cyclically attended by the server. The buffer could store the messages that it wants to send, in such cases, messages are not relevant after reconfiguring, and thus keeping these messages can slow down the network and making hosts reconfigure when they do not need to. We therefore considered two different versions of the network: one where the host does not do anything about these messages (No\_Reset) and the another where the host deletes these messages when it decides to choose a new IP address (Reset). 

We chose the property that measures the probability of not choosing a fresh address by time T. This property is given in PRISM as follows:
\begin{flushleft}
	\centering
	$Pmax=?[!(l=4 \wedge ip=2)U (t > T)]$
\end{flushleft}
We test this property using PRISM for both types of network (Reset) and (No\_Reset) for the following values (T = 10; N = 1000;loss = 0.1 and K could vary from 4 to 8). For these values, PRISM renders a probability greater than 0.5.  As a result, we chose the value 0.5 as a threshold.  The property can be rewritten as follows:
\begin{flushleft}
	\centering
	$P\leq 0.5[!(l=4 \wedge ip=2)U (t > T)]$
\end{flushleft}

\begin{table}
	\renewcommand{\arraystretch}{0.5}
	\label{table 1}
	\centering
	\caption{PRISM results for Zeroconf}
	{\begin{tabular}{|c|c|c|c|}
			\hline
			Reset & K & states & transitions\\
			\hline
			true & 4 & 9683 & 15727 \\
			& 6 & 7743 & 11401\\
			& 8 & 7743 & 11401  \\
			\hline
			false & 4 & 59076 & 121265\\
			& 6 & 58937 & 120525  \\
			& 8 & 58937 & 120525  \\
			\hline
		\end{tabular}}
	\end{table}
	
	The PRISM results are shown in Table 1. We notice that the size of the models is very huge with (No\_Reset), comparing to (Reset). Despite that, DiPro renders the same counterexample for all these different configurations with the same execution time. For generating the counterexample, DiPro Explored 24 traces resulting in 121 vertices and 150 edges in 5 seconds for all the configurations. Finally, the counterexample rendered by DiPro consists just of 8 diagnostic paths. 
	
	We pass this counterexample to our algorithm for generating the diagnoses. Our algorithm takes less than 1 second as execution time. The causes generated for this property are as follows: for the right sub-formula, the cause generated is singleton $C0=(t > T)$. For the left sub-formula, the set of causes for not choosing a fresh address: $C1= !(l=4)$ and $C2= !(ip=2)$. Notice that we are facing disjunctive scenario here, which means that either $C1$ (address not in use) is the actual cause or $C2$ (not fresh address) or both of them. Our results show that except the initial state where $ip=1$ (IP address of an abstract host which the concrete host is currently trying to configure), the actual cause for not choosing fresh address within 10 time units is $C1$, which means that we reach states in which there is fresh ip which the concrete host is currently trying to configure but without being used. The number of states from 8 diagnostic paths in which we found theses causes are estimated to be 59 states, this is much less than the states of the model (9683).
	
	Concerning the actions and their blame, in the model we have two main actions causing the non-determinism in such states, which are: 'Reconfigure' denoted $rec$ and 'defend' which is performed by sending an ARP packet and thus this action is denoted in the model by $send$. For the counterexample generated given the previous property, our results show that there is no state in which the host reconfigures, which means that the only action we are dealing with is $send$. As a result, we compute the dB of $send$ action in such states. We found that it has the same degree of blame ($dB=0.25$) in each state it is enabled.
	
	\subsection{CSMA/CD (Carrier Sense, Multiple Access with Collision Detection)}
	CSMA/CD is a protocol for carrier transmission access in Ethernet networks that avoids collision (minimizing simultaneous use of the channel) when Network Interface Card (NIC) tries to send its packet. The protocol is modeled as a probabilistic timed automata (PTA) \cite{CSMA31} and can be reduced to an MDP in order to be analyzed against probabilistic properties by PRISM. The model in PRISM consists of three main components or modules, the two senders namely station 1 and station 2 respectively and the third component is the bus or the medium. The protocol functionality is as follows: if a station has a data to send, it listens first to the medium, in case it is free, the stations send the data, otherwise (bus is busy), it repeats the process after random amount of time. If there is a collision the station attempts to retransmit the packet where the scheduling of the retransmission is determined by a truncated binary exponential backoff process. The complete model is available in the PRISM benchmark suite under MDPs section \cite{CSMA32}. 
	
	We chose the property that estimates the maximum probability of all stations sending successfully before a collision with max backoff. This property is given as follows:
	
	\begin{flushleft}
		\centering
		$Pmax=? [ !"collision\_max\_backoff" U "all\_delivered" ]$
	\end{flushleft}
	
	We tested this property using PRISM for the following values: N=2 and K=2, N=2 and K=4 respectively, where N represents the number of stations and K represents the exponential backoff limit. For these values, PRISM renders a probability greater than 0.7.  As a result, we chose the value 0.7 as a threshold.  The property can be rewritten as follows:
	\begin{flushleft}
		\centering
		$P\leq 0.7[ !"collision\_max\_backoff" U "all\_delivered" ]$
	\end{flushleft}
	
	Where $collision\_max\_backoff$ and $all\_delivered$ are defined as follows:
	$collision\_max\_backoff$\\ $= (cd1=K \& s1=1 \& b=2)|(cd2=K \& s2=1 \& b=2)$, $all\_delivered = (s1=4\&s2=4)$
	. The variables $cd1$ and $cd2$ refer to collision counters for both stations where $K$ represents the backoff limit , $s1$ and $s2$ refer to the state of the stations where $s1=1,s2=1$ indicate that the stations are transmitting data, and finally $b$ refers to the state of the bus where $b=2$ indicates that there is a collision. 
	
	\begin{table}
		\centering
		\caption{PRISM Results}
		{\begin{tabular}{|c|c|c|c|}
				\hline
				K & States & Transitions\\
				\hline
				2 & 1083 & 1282 \\
				4 & 7958 & 10594 \\
				\hline
			\end{tabular}}
			
			\hfill
		\end{table}
		\begin{table}
			\centering
			\caption{DiPro Results}
			{\begin{tabular}{|c|c|c|c|c|}
					\hline
					K & States & Transitions &  Time Construction (Sec) & Diagnostic Paths\\
					\hline
					2 & 1037 & 1276  & 6 sec & 134\\
					4 & 4222 & 5201 & 16 sec & 324\\
					\hline
				\end{tabular}}
			\end{table}
			
			Table 2 shows the size of the model by PRISM. Table 3 shows the states and transition explored while searching for the counterexample and the time required for its construction by DiPro. We notice that DiPro nearly explored all the model to generate the counterexample for K=2, whereas for K=4, DiPro explores nearly half of the model. For k=2, the counterexample generated consists of 134 diagnostic paths, and for k=4 the counterexample generated consists of 324 diagnostic paths. We pass both counterexamples to our algorithm for generating the diagnoses.
			
			\begin{table}
				\renewcommand{\arraystretch}{0.5}
				\label{table_example}
				\centering
				\caption{Execution results of our algorithm}
				{\begin{tabular}{|c|c|c|c|}
						\hline
						K & Causes & Causes Probabilities &  Execution Time (Sec)\\
						\hline
						2 & 616 & 98 & 3 sec\\
						4 & 923 & 47 & 8 sec\\
						\hline
					\end{tabular}
				}
			\end{table}
		
	Table 4 shows the execution results of our algorithm. The second column represents the number of causes generated from the counterexample with respect to the states, while the third column shows the number of classes of causes probabilities. The results show that the time taken for computing the causes is less than the time taken for generating the counterexample.
			
	The causes generated for this property are as follows: For the right sub-formula, the cause generated is a conjunct $C0=(s1=4\&s2=4)$. For the left sub-formula, the set of causes for not facing a collision with max backoff are: $C1=\neg(cd1=K)$, $C2=\neg (s1=1)$, $C3=\neg(b=2)$,$C4= \neg(cd2=K)$
	$C5=\neg(s2=1)$. For both values of K, our results show that there are causes that share the same probability, as we mentioned before that most responsible cause may not be unique. In addition, for both values of K, the most responsible causes are the same. Among all the causes, we found that the most responsible causes for not facing collision are $C1$ and $C4$, where the states in which these causes are found have the highest probability. 
	
	Concerning the actions and their blame, we found that $send1$ and $send2$ have the most blame, where $send1$ and $send2$ lead to the most responsible causes $C1$ and $C4$. For $send1$, the first transition that it enables is represented by the valuations $b=1$ (bus active) and $s1=1$ (station 1 is transmitting). Given these information we will be able to identify the commands that contributed the most to the violation, which are $[send1] (b=0) -> (b'=1)$ (line 35) in the bus module from a side, and from the side of station, the command is $[send1] (s1=0) -> (s1'=1) \& (x1'=0)$ (line 82). For $send2$, the first transition that it enables is represented by the valuations $b=2$ (bus busy - collision) and $s2=1$ (station 2 is transmitting). Given these information we will be able to identify the commands, which are $[send2] (b=1|b=2) \& (y2<sigma) -> (b'=2)$ (line 40) in the bus module from a side, and from the side of station, the command is $[send2] (s2=0) -> (s2'=1) \& (x2'=0)$ (line 82). Then, based on the second transition enabled by $send2$, which is represented by the valuations  $b=1$ (bus active) and $s2=1$ (station 2 is transmitting) and leads to a critical state, the commands to check are $[send2] (b=0) -> (b'=1)$ (line 36) in the bus module, and $[send2] (s2=3) \& (x2=slot) \& (bc2=0) -> (s2'=1) \& (x2'=0)$ (line 99) in the station module respectively. Finally, based on the second transition enabled by $send1$, which is represented by the valuations $b=2$ (bus busy - collision) and $s1=1$ (station 1 is transmitting), the commands to check are: $[send1] (b=1|b=2) \& (y1<sigma) -> (b'=2)$ (line 39)in the bus module, and $[send1] (s1=3) \& (x1=slot) \& (bc1=0) -> (s1'=1) \& (x1'=0)$ (line 99) in the station module respectively. The other states in which $C1$ and $C4$ are not the causes, we find that $C3$ is the most responsible cause. So by defining the transition leading to it, and under which action is enabled, we could also map to the commands related and analyze the rest of the model driven by the diagnoses. So as we see, given the diagnoses generated by our algorithm, it would be easy to go-back to the model, which is given in PRISM guarded command language, and identify the commands that contributed the most to the violation of the probability threshold. By performing some changes on these commands in order, the probability as estimated by PRISM goes below the probability threshold, and thus the probabilistic property will be satisfied.
	
	In general, the results presented here report that our method has a good execution time and thus it can be adapted in practice. Comparing to the execution time taken for generating the counterexample, the execution time of our method is remarkably lower. Concerning the number of diagnoses generated, our method outputs low number of causes comparing to the size of the model, in addition our method outputs the most responsible causes first to the user, which could help the user to find the source of the error rapidly without reading all the causes generated. Along the results, we found that many causes could share the same probability, which means that many causes are found in the same set of paths of the counterexample, this could mean that there is such dependence between variables where the values are changing together. This information is very important for debugging, because it could help user even to locate the line responsible for error in the model itself given the action leading to it. Concerning actions we notice that such actions from the model might be ignored since they do not exist under the scheduler resolving the non-determinism, for instance for Zerconf protocol, we found the only action to be analyzed is \textit{send}, the action \textit{reconfigure} is not concerned, and that means that the search space is limited. It is evident that relaying on small number of actions also facilitates the debugging. Concerning the structure of the probabilistic property to be analyzed, our method has more importance when we face a disjunctive scenario, because in disjunctive scenario we can not be sure about the variables causing the violation, and here where responsibility measure plays the major role. Finally, the commands as identified in the PRISM modules with respect to the diagnoses generated do not necessarily follow the same order in the module, for instance  the command $[send2] (b=1|b=2) \& (y2<sigma) -> (b'=2)$ (line 40) has been reported before the command $[send2] (b=0) -> (b'=1)$ (line 36). We also notice that we were able to deal with the parallel composition of the modules, for instance  $[send1] (b=0) -> (b'=1)$ in the bus module and $[send1] (s1=0) -> (s1'=1) \& (x1'=0)$ in the station module.

	\section{Conclusion and Future Works}
	In this paper we have shown how the notions of causality, responsibility and blame can be useful in the context of probabilistic counterexamples of MDPs. Due to the probabilistic nature of the causality model, we introduced the definition of most responsible cause and the action with most blame. We showed that delivering the causes/actions with respect to their responsibilities/blame stands as a good debugging method that guides the user through large counterexamples, and thus it could help us to identify the commands responsible for the violation of the probabilistic formula.
	
	As future works, we plan to deliver a debugging tool that generates the diagnoses graphically for all kinds of Markov models. Furthermore, we aim to integrate our method in the model checking process itself, in order to locate the commands in PRISM code directly without depending on the counterexample generated.

\bibliographystyle{eptcs}
\bibliography{CrestBib}
\end{document}